\documentclass[conference]{IEEEtran}
\IEEEoverridecommandlockouts
\usepackage{cite}
\usepackage{amsmath,amssymb,amsfonts}
\usepackage{amsthm}
\usepackage{graphicx}
\usepackage{textcomp}
\usepackage{booktabs}
\usepackage{xcolor}
\newtheorem{definition}{Definition}
\usepackage{subfigure}
\usepackage[capitalize]{cleveref}
\usepackage[ruled,lined,linesnumbered]{algorithm2e}
\usepackage{algpseudocode}
\newtheorem{theorem}{Theorem}
\newtheorem{assumption}{Assumption}

\usepackage{float}
\def\BibTeX{{\rm B\kern-.05em{\sc i\kern-.025em b}\kern-.08em
    T\kern-.1667em\lower.7ex\hbox{E}\kern-.125emX}}
\begin{document}
\makeatletter
\newcommand{\linebreakand}{%
  \end{@IEEEauthorhalign}
  \hfill\mbox{}\par
  \mbox{}\hfill\begin{@IEEEauthorhalign}
}
\makeatother

\title{
Efficient Federated Unlearning with Adaptive Differential Privacy Preservation
}
\author{\IEEEauthorblockN{Yu Jiang}
\IEEEauthorblockA{\textit{Nanyang Technological University} \\
Singapore \\
yu012@e.ntu.edu.sg} \\
\and
\IEEEauthorblockN{Xindi Tong}
\IEEEauthorblockA{\textit{Nanyang Technological University} \\
Singapore \\
to0001di@e.ntu.edu.sg} \\
\and
\IEEEauthorblockN{Ziyao Liu}
\IEEEauthorblockA{\textit{Nanyang Technological University} \\
Singapore \\
liuziyao@ntu.edu.sg} \\
\and
\IEEEauthorblockN{Huanyi Ye}
\IEEEauthorblockA{\textit{Nanyang Technological University} \\
Singapore \\
huanyi001@e.ntu.edu.sg} 
\and
\IEEEauthorblockN{Chee Wei Tan}
\IEEEauthorblockA{\textit{Nanyang Technological University} \\
Singapore\\
cheewei.tan@ntu.edu.sg} \\
\and
\IEEEauthorblockN{Kwok-Yan Lam}
\IEEEauthorblockA{\textit{Nanyang Technological University} \\
Singapore\\
kwokyan.lam@ntu.edu.sg}
}

\maketitle

\begin{abstract}
Federated unlearning (FU) offers a promising solution to effectively address the need to erase the impact of specific clients' data on the global model in federated learning (FL), thereby granting individuals the ``Right to be Forgotten". The most straightforward approach to achieve unlearning is to train the model from scratch, excluding clients who request data removal, but it is resource-intensive. 
Current state-of-the-art FU methods extend traditional FL frameworks by leveraging stored historical updates, enabling more efficient unlearning than training from scratch. However, the use of stored updates introduces significant privacy risks. Adversaries with access to these updates can potentially reconstruct clients' local data, a well-known vulnerability in the privacy domain. While privacy-enhanced techniques exist, their applications to FU scenarios that balance unlearning efficiency with privacy protection remain underexplored. To address this gap, we propose FedADP, a method designed to achieve both efficiency and privacy preservation in FU. Our approach incorporates an adaptive differential privacy (DP) mechanism, carefully balancing privacy and unlearning performance through a novel budget allocation strategy tailored for FU. FedADP also employs a dual-layered selection process, focusing on global models with significant changes and client updates closely aligned with the global model, reducing storage and communication costs. Additionally, a novel calibration method is introduced to facilitate effective unlearning. Extensive experimental results demonstrate that FedADP effectively manages the trade-off between unlearning efficiency and privacy protection.

\end{abstract}


\section{Introduction}

Federated learning (FL) has been established as a critical paradigm in the field of distributed machine learning, facilitating collaborative model training across multiple participants without sharing their private data \cite{mcmahan2017federated, boyd2011distributed, liu2023long, liu2022efficient, liu2024dynamic}. This approach effectively addresses the requirement for data privacy. As attention to personal data sovereignty intensifies, regulations such as the General Data Protection Regulation (GDPR) \cite{regulation2018general} and the California Consumer Privacy Act (CCPA) \cite{goldman2020introduction} have been introduced, granting individuals the ``Right to be Forgotten" and mandating the removal of personal data upon request. These regulations impose new compliance demands, creating an urgent need for mechanisms capable of effectively removing the influence of specific data from trained models. This need has catalyzed the development of federated unlearning (FU) \cite{liu2024survey, jiang2024towards, cao2015towards, bourtoule2021machine}, an emerging area of research dedicated to facilitating data removal within federated learning frameworks.

The most straightforward way to achieve unlearning is to retrain the model from scratch, excluding clients who request data removal; however, this method is resource-intensive. Current state-of-the-art FU methods primarily build upon traditional FL frameworks by incorporating new features or functions, such as FedRecover \cite{cao2023fedrecover}, in which the server retains historical updates or models from clients for future unlearning. While using stored information enhances unlearning efficiency, it incurs high data communication and storage costs. Besides, storing updates or models introduces vulnerabilities, such as membership inference attack \cite{shokri2017membership} and model inversion attack \cite{fredrikson2015model}. These vulnerabilities create opportunities for adversaries to infer and reconstruct personal data, thereby contradicting the fundamental goal of privacy preservation in FL.
While privacy-enhanced techniques exist, their applications to FU scenarios that balance unlearning efficiency with privacy protection remain underexplored.

To address the challenge of balancing efficiency and privacy preservation in federated unlearning, we propose FedADP, which leverages an adaptive differential privacy (DP) mechanism \cite{dwork2006differential,bureau2021disclosure, dwork2014algorithmic} to carefully balance privacy and unlearning performance. 
Specifically, the server applies an exponential equation for adaptive budget allocation, enabling control over noise levels in subsequent training rounds to balance model performance. Based on the allocated privacy budget, clients apply the Gaussian mechanism to protect sensitive information in their updates before sharing them with the server, significantly reducing the risk of data leakage.
Furthermore, FedADP incorporates a dual-layered selection strategy for historical information, focusing on global models with substantial changes and prioritizing updates closely aligned with the global model, which significantly reduces communication and storage costs.

After target clients submit unlearning requests, the server and the remaining clients initiate the unlearning process following the standard federated learning workflow. During unlearning, we introduce a novel calibration mechanism that facilitates unlearning by calibrating the magnitude and direction of model updates from historical data with new model updates. The comprehensive FedADP framework is shown in \cref{fig:Workflow}. We provide a theoretical analysis of how privacy affects convergence and establish an upper bound on convergence while ensuring DP for all clients. Our experiments, conducted across various datasets and attack scenarios, demonstrate FedADP's effectiveness in maintaining high accuracy and strong resistance to attacks. We also evaluate its performance with different aggregation methods and examine how selection strategies and privacy budgets influence unlearning. FedADP has been proven to effectively manage the trade-off between unlearning efficiency and privacy protection.

\begin{figure*}[t]
    \centering
    \includegraphics[width=\linewidth]{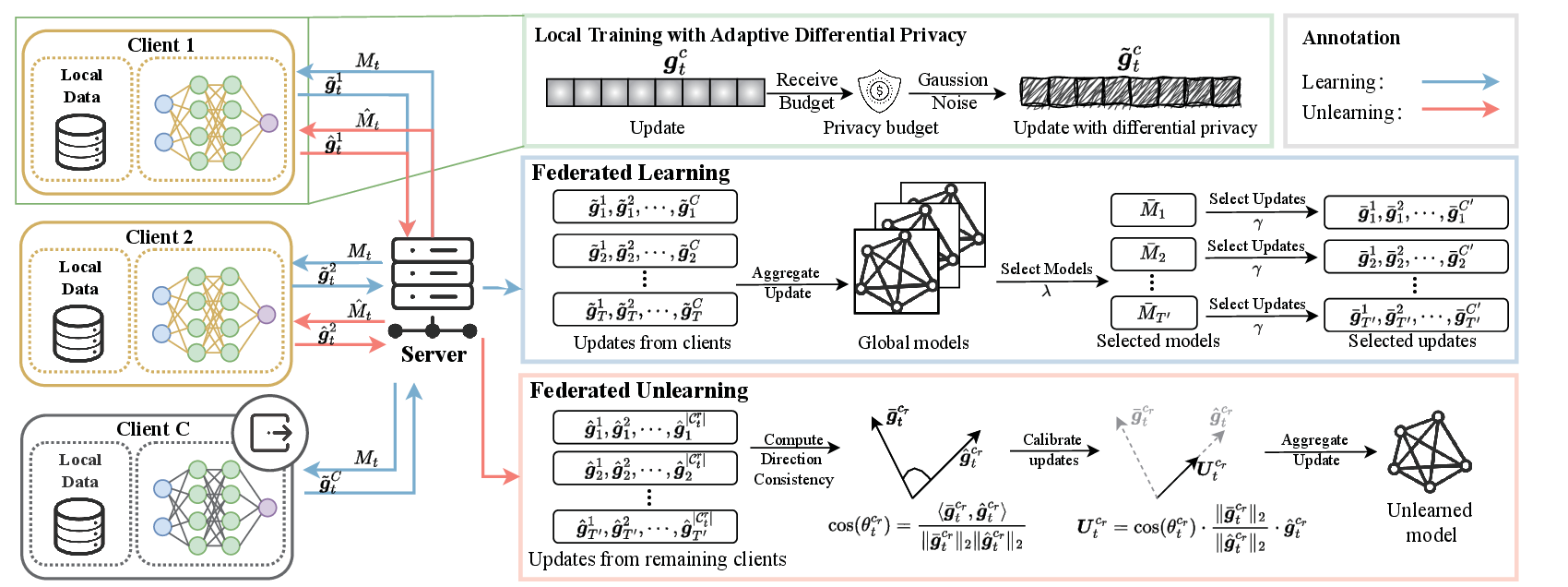}
    \caption{The workflow of FedADP from FL to FU. In FL, clients train local models by adaptively applying differential privacy while the server adopts model and update selection. In FU, the server calibrates historical information for unlearning.}
    \label{fig:Workflow}
\end{figure*}

\textbf{Our Contributions.}
The main contributions of this work are outlined below.
\begin{enumerate}
\item We propose a privacy-preserving approach tailored for FU, in which the server guides clients to adaptively apply DP through an exponential equation for budget allocation, enabling clients to add Gaussian noise based on the allocated privacy budget.

\item We introduce an efficient federated unlearning method that selectively stores historical global models and critical updates, significantly reducing communication and storage costs. We also design a novel calibration mechanism, facilitating the effective unlearning.

\item Through extensive experimentation on diverse datasets, attack scenarios, and aggregation methods, we validate FedADP's performance, demonstrating its capability to balance unlearning effectiveness with privacy protection.
\end{enumerate}



\section{Related Work}\label{sec:relatedwork}
\subsection{Federated Unlearning}
FU aims to eliminate the influence or specific data of target clients requesting data removal from a global model. Some approaches achieve highly effective unlearning using historical information, such as FedRecover \cite{cao2023fedrecover}. Despite their effectiveness, these methods have been critiqued for the potential risk of information leakage and the significant costs associated with data storage and communication. 
Although studies such as SFU \cite{li2023subspace}, IFU \cite{fraboni2022sequential}, and FedRecovery \cite{zhang2023fedrecovery} have considered privacy issues, they significantly increase resource consumption or exhibit unstable unlearning performance.
On the one hand, SFU requires the remaining clients to create representation matrices to send to the server, adding randomness to the process. The server then combines these matrices using Singular Value Decomposition (SVD) to update the global model, which introduces substantial computational overhead. On the other hand, IFU and FedRecovery both employ historical data by adding perturbations. IFU applies disturbances to a client's contributions, while FedRecovery adds Gaussian noise to a weighted sum of gradient residuals. However, accurately calculating a client's contributions or the weighted sum of gradient residuals is crucial; otherwise, inaccuracies detrimentally affect model utility.
In contrast, FedADP balances privacy protection with computational and storage efficiency, without compromising the performance of the global model.

\subsection{Aggregation Methods}
Different FL methods essentially use different aggregation rules.
FedAvg~\cite{mcmahan2017communication} is a commonly employed aggregation method in federated learning where the global model is updated by weighted averaging the model updates from clients. Each client trains locally using its data and computes updates, which the server aggregates using a weighted average based on the number of data samples at each client.
Trimmed Mean~\cite{xie2018generalized} is a robust aggregation technique that reduces the influence of outliers by removing a certain percentage of the highest and lowest values from the list of client updates before computing the mean. This approach helps mitigate the impact of clients with anomalous or corrupted updates, thereby enhancing the robustness of the global model.
Median~\cite{yin2018byzantine} is another robust aggregation method that selects the middle value from the sorted list of client updates. By choosing Median, a more stable aggregation is achieved, especially when dealing with noisy or adversarial client updates, as it is less sensitive to outliers compared to the mean.

\section{Preliminaries}\label{sec:preliminaries}
\subsection{Federated Learning}
In FL, the architecture comprises two primary entities: a group of $C$ clients and a central server. The server coordinates the training of a global model, $M$, using datasets distributed across these clients, each denoted by $D_c$ from client $c \in \mathcal{C}$, forming the complete dataset $D$. The objective is to minimize a loss function, $\min_M \mathcal{L}(D;M)= \min_M \sum_{c=1}^C \mathcal{L}(D_c;M)$, where $\mathcal{L}$ represents the empirical loss function, such as cross-entropy loss. For simplicity, we denote $\mathcal{L}_c(M)=\mathcal{L}(D_c;M)$.
The overall process involves iterative training rounds. 
Specifically, at the \( t \)-th round, the server broadcasts the current global model \( M_t \) to the clients, who then independently train local models on their datasets to compute updates \( \boldsymbol{g}_t^c = \frac{\partial \mathcal{L}(D_c; M_t)}{\partial M_t} \) using stochastic gradient descent. These updates are sent back to the server. The server aggregates them according to a specific rule \( \mathcal{A} \) and updates the global model as \( M_{t+1} = M_t - \eta \mathcal{A}(\boldsymbol{g}_t^1, \boldsymbol{g}_t^2, \dots, \boldsymbol{g}_t^C) \), where \( \eta \) is the learning rate. The server then distributes the updated global model to clients for the next round of training.
This cycle repeats until a predefined convergence criterion is met or the training goal is achieved, culminating in the final model. This process highlights the collaborative nature of FL, which aims to enhance model learning while maintaining the privacy of clients' datasets.


\subsection{Differential Privacy}
Differential privacy is a framework that adds controlled noise to data, ensuring individual privacy while allowing statistical analysis, by quantifying and limiting the privacy risk associated with data release or processing.
\begin{definition}[($\epsilon$, $\delta$)-differential privacy \cite{dwork2006differential}]
For any two datasets $ D $ and $ D' $ that differ in at most one element, and for all subsets of outputs $ S $, a mechanism $ \mathcal{M} $ satisfies $(\epsilon, \delta)$-differential privacy if
\begin{equation}
    \Pr[\mathcal{M}(D) \in S] \leq e^\epsilon \cdot \Pr[\mathcal{M}(D') \in S] + \delta,
\end{equation}
where $\mathcal{M}$ is a randomized mechanism, and the parameters $ \epsilon $ and $ \delta $ quantify the privacy guarantee, with $\epsilon$ controlling the privacy loss and $\delta$ representing the probability of exceeding this privacy loss.
\end{definition}

The concept of sensitivity plays a critical role in differential privacy. It quantifies the potential impact of a single individual’s data on the output of a function, thereby guiding the calibration of the noise needed to achieve privacy. Formally, the sensitivity $\Delta f$ of a function $f: D \rightarrow \mathbb{R}^n$, operating on a dataset $D$, is defined as
\begin{equation}\label{eq:sensitivity}
    \Delta f  = \max_{D, D'} \| f(D) - f(D') \|,
\end{equation}
where the maximum is taken over all pairs of datasets $D$ and $D'$ that differ by exactly one element.

The Gaussian noise mechanism, a common approach in differential privacy, is described by the equation
\begin{equation}
    \mathcal{M}(D) \triangleq f(D) + \mathcal{N}(0, \sigma^2 I),
\end{equation}
where $\mathcal{N}(0, \sigma^2 I)$ denotes the Gaussian noise with mean 0 and covariance matrix $\sigma^2 I$, and $I$ is the identity matrix. The standard deviation $\sigma$ of the Gaussian noise is determined based on the sensitivity $\Delta f$ and the privacy parameters $\epsilon$ and $\delta$, ensuring that the mechanism adheres to the $(\epsilon, \delta)$-differential privacy criteria.


\subsection{Threat Model}
In this work, we assume that information transferred over the network is securely protected during transmission. The adversary could either be an ``honest-but-curious" central server or a client within the system, both assumed to have access to stored model updates. The adversary’s primary goal is to exploit these model updates to recover sensitive training data. This can be achieved through model inversion attacks, where the adversary attempts to reconstruct individual data points, or through membership inference attacks, where the adversary infers whether a specific sample was part of the training dataset.

\section{Design of FedADP}\label{sec:design}
\subsection{Overview}

Our proposed FedADP scheme enhances privacy preservation and unlearning efficiency through two main strategies: (i) applying adaptive differential privacy with a dual-layered selection of critical historical data in FL, and (ii) calibrating historical information in FU. The FedADP methodology is illustrated in \cref{fig:Workflow}. This section provides a detailed explanation of these strategies as part of the comprehensive overview of the FedADP scheme.

\subsection{Design of Adaptive DP and Dual-layered Selection in FL}
\subsubsection{Adaptive Differential Privacy} \label{sec:dp}

Based on our observations, it is necessary to store clients' updates to facilitate unlearning. However, due to the potential risks of information leakage, certain strategies are required to safeguard the privacy of clients.
Differential privacy \cite{wei2020federated, noble2022differentially, hu2020personalized}, a prevalent method in federated environments, quantifies and controls privacy loss through mathematical frameworks. DP provides a mechanism to protect individual data privacy by setting an upper bound on the information that can be inferred about any specific entry. Clients add controlled noise to their gradients, ensuring that updates used for updating the global model do not reveal individual contributions, thereby protecting sensitive client information. While established methods exist for applying DP in federated learning, effective approaches for adapting DP to federated unlearning remain underexplored.

To address this gap, we adopt an adaptive differential privacy strategy tailored for unlearning. Specifically, the server instructs the clients to add noise according to a dynamic privacy budget. The budget allocation for each iteration in our training process follows an exponential equation that utilizes the mathematical constant $e$ and the change in training loss between two consecutive rounds. This equation is designed to adjust the budget based on the observed performance improvements.
The budget for the next round, denoted as $\epsilon_{t+1}$, is calculated using the following formula:
\begin{equation}
    \epsilon_{t+1} = \min(\max(\epsilon_t \cdot e^{\Delta_\text{loss}}, \epsilon_{\min}), \epsilon_{\max}),
\end{equation}
where $\epsilon_t$ represents the budget allocated in the current round, and $\Delta_\text{loss} =| \text{loss}_{t-1} - \text{loss}_{t} |$ is the difference in training loss between the previous round and current round. 
This formula allocates an increased budget to address significant variations in training loss, thereby supporting a stable training process. The exponential function $e^{\Delta_\text{loss}}$ scales the budget to accommodate substantial shifts in training loss, promoting more precise and stable model tuning.


The use of the exponential equation enables an adaptive approach to budget allocation, where reduced noise levels in subsequent training rounds lead to enhanced model performance. This adjustment aims to balance the need for computational resources with the goal of optimizing the training process. To prevent excessive budget fluctuations, the formula includes boundary constraints, $\epsilon_{\min}$ and $\epsilon_{\max}$. These boundaries ensure that the budget does not fall below a minimum limit or exceed a maximum threshold, thereby maintaining stability and control over the total budget allocated in each training round. This approach aids in managing resources effectively while progressively refining the model's performance.

Then, we employ the Gaussian mechanism to protect sensitive information in the clients' updates. Given the variability in update sizes across clients, norm clipping is essential, which constrains each client's update to a predefined threshold $ S $, i.e.,
$
\dot{\boldsymbol{g}}_t^c = {{\boldsymbol{g}}_t^c} / {\max\left(1, \frac{\|{\boldsymbol{g}}_t^c\|_2}{S}\right)},
$
where $ {\boldsymbol{g}}_t^c $ is the original update from client $ c $ at round $ t $, and $ \dot{\boldsymbol{g}}_t^c $ is the clipped update. To incorporate DP, noise is added to the clipped update, with the noise drawn from a Gaussian distribution with mean zero and variance $ \sigma_t^2 S^2 $, resulting in the noised update:
\begin{equation}
\tilde{\boldsymbol{g}}_t^c = \dot{\boldsymbol{g}}_t^c + \mathcal{N}(0, \sigma_t^2 S^2).
\end{equation}
The standard deviation $ \sigma_t $ is calculated based on the privacy parameters $ \epsilon_t $ and $ \delta $, and the sensitivity dictated by the clipping threshold $ S $, using the formula:
\begin{equation}
\sigma_t = \frac{S \cdot \sqrt{2 \ln(1.25/\delta)}}{\epsilon_t}.
\end{equation}
This ensures that the added Gaussian noise is sufficient to meet the ($ \epsilon $, $ \delta $)-differential privacy criteria. By scaling $ \sigma_t $ with $ S $, the noise magnitude is appropriately adjusted to the maximum allowed update size, thereby maintaining the utility of the aggregated updates while safeguarding clients' privacy. 

\subsubsection{Dual-layered Selection}
In traditional FL, clients send full model updates, leading to high communication costs and memory challenges for servers. To address this, we make selections at two levels: (i) global models with significant changes, and (ii) updates that are more aligned with the models.

\textbf{Global model selection.}
It is reasonable for the server to prioritize storing global models with higher convergence rates, as the convergence speed of the global model varies across different rounds. A high convergence rate indicates significant changes in the global model, making these models essential for simplifying future unlearning processes by retaining only the most impactful changes, thereby reducing unnecessary storage.
Cosine similarity is a powerful tool for measuring the angular similarity between two vector representations, which is valuable in machine learning for comparing outputs of different models \cite{sattler2020byzantine, wang2021edge,cao2021fltrust}. By assessing the cosine similarity between successive model updates, we quantify how closely these models align. When combined with the ReLU function, which selectively retains the positive components of the similarity measure while discarding the negative ones, this method becomes particularly effective at highlighting significant updates between models.
Consequently, a lower cosine similarity after ReLU indicates a greater divergence between models, helping to identify which models should be prioritized for retention.

To implement selective storage using cosine similarity combined with the ReLU function, we instruct the server to evaluate the rate of alignment or changes in the stored models. Let $ M_t $ represent the global model at the $ t $-th round. The server then computes the cosine similarity between the current model $ M_{t} $ and its predecessor $ M_{t-1} $ as follows:
\begin{equation}
   \text{cos}(M_{t}, M_{t-1}) = \frac{\langle M_{t}, M_{t-1} \rangle}{\|M_{t}\|_2 \|M_{t-1}\|_2},
\end{equation}
where $\left\langle  , \right\rangle$ is the scalar product operator and $\Vert \cdot \Vert$ represents the $\ell_2$ norm of a vector.
The degree of alignment between the two successive models is refined by applying the ReLU function to the cosine similarity, emphasizing positive changes:
\begin{equation}
    d(M_{t}, M_{t-1})=\text{ReLU}(\text{cos}(M_{t}, M_{t-1})).
\end{equation}
Subsequently, the server selects $ \lambda $ (percentage) of the models with the least alignment $d(M_{t}, M_{t-1})$, prioritizing those with the most substantial changes to optimize storage and facilitate efficient unlearning. This selection method does not require waiting for the completion of training; instead, it retains global models at each stage when the loss decreases to $ \beta $ (percentage) of the previous stage’s value, allowing for timely model evaluation and storage optimization based on performance improvements.
Therefore, when an unlearning request is issued, the final set of stored models is denoted as $ \{\bar{M}_t\}^{T'} =\{\bar{M}_t|t=1,\dots,T'\}$, where $ T' = \lambda T $ represents the total number of global models ultimately stored, $ \lambda $ is the global model selection ratio, and $ T $ is the total training rounds in the federated learning process. Along with these models, the corresponding updates are denoted by $\{G_t\}^{T'}$, where each $G_t$ consists of the gradients from all $C$ clients at round $t$, i.e., $ G_t = \{{\tilde{\boldsymbol g}}_t^c\}^C =\{{\tilde{\boldsymbol g}}_t^c |  c=1,\dots, C\} $. Then, we let $\boldsymbol{G}_t = \mathcal{A}(\tilde{\boldsymbol g}_t^1, \tilde{\boldsymbol g}_t^2, \dots, \tilde{\boldsymbol g}_t^c )$ denote the aggregated update for the global model at the $t$-th round.



\textbf{Model update selection.} Similarly, to further optimize storage, we select model updates from clients that are more closely aligned with the global model. Updates that are closer to the global model indicate a greater contribution and have a more significant impact on the global models, making them more important. We use cosine similarity to quantify the alignment of the clients' model updates, with a higher cosine similarity indicating closer alignment and a greater contribution.

After selecting the global models, the server utilizes the cosine similarity to compare the alignment between the model update from client with the aggregated update:
\begin{equation}
    s(\tilde{{\boldsymbol g}}_t^c,\boldsymbol{G}_t) = \text{cos}(\tilde{{\boldsymbol g}}_t^c,\boldsymbol{G}_t)=\frac{\langle \tilde{{\boldsymbol g}}_t^c, \boldsymbol{G}_t \rangle}{\|\tilde{{\boldsymbol g}}_t^c\|_2 \|\boldsymbol{G}_t\|_2}. 
\end{equation}
Then, the server selects $ \gamma $ (percentage) of the updates with the highest alignment $s(\tilde{{\boldsymbol g}}_t^c,\boldsymbol{G}_t)$, representing the most important model updates. The final set of selected model updates for each round is therefore given by $\bar G_t = \{\bar{\boldsymbol g}_t^c\}^{C'} = \{ \bar{\boldsymbol g}_t^c | c=1, \dots, C' \}$, where $C'=\gamma C$ and $\gamma$ is the update selection ratio. Meanwhile, the client sets providing these updates are denoted as $\{\mathcal{C}_t\}^{T'}$, who will participate in the subsequent unlearning process. A detailed description of adaptive DP and dual-layered selection in FL within the FedADP scheme is provided in Algorithm \ref{alg:flinFedADP}.

\begin{algorithm}[t]
\SetAlgoNoEnd
\caption{Adaptive Differential Privacy and Dual-layered Selection in FL of FedADP Scheme}
\label{alg:flinFedADP}
\KwIn{Federated learning client $c \in \mathcal{C}$ with dataset $D_c$, initial model $M_0$, global model selection ratio $\lambda$, model update selection ratio $\gamma$, percentage of loss reduction $\beta$, norm clipping threshold $S$, differential privacy parameters $\epsilon$ and $\delta$, maximum training budget for each round $\epsilon_{\max}$, minimum training budget for each round $\epsilon_{\min}$}
\KwOut{Final storage: $\{\bar{M}_t\}^{T'}$, $\{\bar G_t\}^{T'}$}
\SetKwFunction{FL}{}
\SetKwProg{Fn}{}{Server executes:}{\KwRet}
\Fn{\FL}{
Initialize the global model $M_0$; \\
Initialize the privacy budget $\epsilon_0$;\\
Initialize $\text{prev\_loss} \leftarrow \mathcal{L}(M_{0})$;\\
Initialize the temporary storage $\mathcal{T}_1$, $\mathcal{T}_2$;\\
Send global model $M_t$ and budget $\epsilon_t$ to all clients; \\
}
\SetKwProg{Fn}{}{Clients execute:}{\KwRet}
\Fn{\FL}{
\For{all client $c$}{
Compute update: $g_t^c \leftarrow \frac{\partial \mathcal{L}(D_c; M_t)}{\partial M_t }$; \\
Compute standard deviation of the noise: $\sigma_t \leftarrow \frac{S \cdot \sqrt{2 \ln(1.25/\delta)}}{\epsilon_t}$;\\
Add Gaussian noise: $\tilde {\boldsymbol g}_t^c  \leftarrow {{\boldsymbol{g}}_t^c} / {\max\left(1, \frac{\|{\boldsymbol{g}}_t^c\|_2}{S}\right)} + \mathcal{N}(0, \sigma_t^2 S^2)$; \\
}
}
\SetKwProg{Fn}{}{Server execute:}{\KwRet}
\Fn{\FL}{
Aggregate updates: $ \boldsymbol{G}_t  \leftarrow \mathcal{A}(\tilde{\boldsymbol g}_t^1, \tilde{\boldsymbol g}_t^2, \dots, \tilde{\boldsymbol g}_t^c )$; \\
Update model: $M_{t+1} \leftarrow M_t - \eta \boldsymbol{G}_t$; \\
Storage temporarily: $\mathcal{T}_1 \leftarrow \mathcal{T}_1 \cup \{M_{t+1}\}$, $\mathcal{T}_2 \leftarrow \mathcal{T}_2 \cup \{ G_t\}$;\\
\If{$\mathcal{L}(M_{t+1}) \leq (1-\beta) \cdot \text{prev\_loss}$}{
\For{temporary storage $\mathcal{T}_1$}{
Store selected model: $\{\bar{M}_t\}  \leftarrow \text{min\_alignment}(\mathcal{T}_1, \lambda)$;\\}
\For{temporary storage $\mathcal{T}_2$}{
Store selected update: $\{\bar{G}_t\}  \leftarrow \text{max\_alignment}(\mathcal{T}_2, \gamma)$;\\
}
Empty temporary storage $\mathcal{T}_1$, $\mathcal{T}_2$;
}
Update loss difference: $\Delta_{\text{loss}} \leftarrow   | \text{loss}_{t-1} - \text{loss}_{t} | $;\\
Update next round budget: $\epsilon_{t+1} \leftarrow \min(\max(\epsilon_t
\ast e^{\Delta_\text{loss}},\epsilon_{\min}),\epsilon_{\max})$;\\
}
\KwRet{$\{\bar{M}_t\}^{T'}$, $\{\bar G_t\}^{T'}$}
\end{algorithm}

\begin{algorithm}[t]
\SetAlgoNoEnd
\caption{Calibration in FU of FedADP Scheme}
\label{alg:fuinFedADP}
\KwIn{Federated learning client $c \in \mathcal{C}$ with dataset $D_c$, unlearning client $c \in \mathcal{C}_u$, stored models $\{\bar{M}_t\}^{T'}$, stored updates $\{\bar G_t\}^{T'}$, client sets providing stored updates $\{\mathcal{C}_t\}^{T'}$}
\KwOut{Unlearned model $\hat M$}
\SetKwFunction{FU}{}
\SetKwProg{Fn}{}{Server executes:}{\KwRet}
\Fn{\FU}{
Reinitialize the global model $\hat M_0 \leftarrow \bar{M}_0$; \\
Send global model $\hat M_t$ to the remaining clients $c_r \in \mathcal{C}_t^r $ $ ( \mathcal{C}_t^r =
\mathcal{C}_t \setminus \mathcal{C}_u)$; \\}
\SetKwProg{Fn}{}{Remaining clients execute:}{\KwRet}
\Fn{\FU}{
\For{all remaining clients $c_r  \in \mathcal{C}_t^r$}{
Compute update: $\hat g_t^{c_r}$;
}}
\SetKwProg{Fn}{}{Server executes:}{\KwRet}
\Fn{\FU}{
Compute direction consistency: $\cos(\theta_t^{c_r}) \leftarrow \frac{\langle \bar{\boldsymbol{g}}_t^{c_r},\hat{\boldsymbol{g}}_t^{c_r} \rangle}{\|\bar{\boldsymbol{g}}_t^{c_r}\|_2 \|\hat{\boldsymbol{g}}_t^{c_r}\|_2}$;\\
Calibrate update: $\boldsymbol U_t^{c_r} \leftarrow \cos(\theta_t^{c_r}) \cdot \frac{\|\bar{\boldsymbol{g}}_t^{c_r}\|_2}{\|\hat{\boldsymbol{g}}_t^{c_r}\|_2} \cdot \hat{\boldsymbol{g}}_t^{c_r}$;\\
Update model: $\hat{M}_{t+1} \leftarrow \hat{M}_t - \eta \boldsymbol{U}_t$;\\
}
\KwRet{$\hat M$}
\end{algorithm}

\subsection{Design of Calibration in FU}
\subsubsection{Calibration of Historical Information}
After target clients $c \in \mathcal{C}_u$ request data erasure, the server and the remaining clients initiate the unlearning task. Here, $\mathcal{C}_u$ denotes the group of clients requesting data unlearning. It is noted that the initial model for unlearning, denoted as $\hat{M}_0$, is derived from the first stored global model, $\bar{M}_0$.

At the $t$-th unlearning round, the server distributes the unlearned global model $\hat{M}_t$ to the remaining clients $c_r \in \mathcal{C}_t^r$, where $\mathcal{C}_t^r = \mathcal{C}_t \setminus \mathcal{C}_u$ according to the client providing the selected updates.
These clients train the global model as in standard FL training and send their updates, $\hat{\boldsymbol{g}}_t^{c_r}$, to the server.
Next, the server calibrates the historical updates based on the new updates from the remaining clients. Specifically, the server calculates the cosine similarity $ \cos(\theta_t^{c_r}) $ between the historical updates and the current updates to assess the direction consistency:
\begin{equation} 
\cos(\theta_t^{c_r}) = \frac{\langle \bar{\boldsymbol{g}}_t^{c_r},\hat{\boldsymbol{g}}_t^{c_r} \rangle}{\|\bar{\boldsymbol{g}}_t^{c_r}\|_2 \|\hat{\boldsymbol{g}}_t^{c_r}\|_2}.
\end{equation}
Then, the server calibrates the update step based on this direction consistency and computes the calibrated update:
\begin{equation}
\boldsymbol U_t^{c_r} = \cos(\theta_t^{c_r}) \cdot \frac{\|\bar{\boldsymbol{g}}_t^{c_r}\|_2}{\|\hat{\boldsymbol{g}}_t^{c_r}\|_2} \cdot \hat{\boldsymbol{g}}_t^{c_r}.
\end{equation}
Once the server collects these calibrated updates, it aggregates them using an aggregation rule:
\begin{equation}
\begin{aligned}
    \boldsymbol{U}_t &= \mathcal{A}(U_t^{1}, U_t^{2}, \ldots, U_t^{|\mathcal{C}_t^r|})  \\
    &= \frac{1}{|D^{-}|} \sum_{c_r \in \mathcal{C}_t^r} |D_{c_r}| \cdot \boldsymbol U_t^{c_r},
\end{aligned}
\end{equation}
where $D^{-}$ represents the dataset from the remaining clients. The server then updates the global model: $\hat{M}_{t+1} = \hat{M}_t - \eta \boldsymbol{U}_t$.
The server and remaining clients repeat these steps until all historical gradients are calibrated, leading to the final unlearned global model $\hat{M}$. A detailed description of the FedADP scheme is provided in Algorithm \ref{alg:fuinFedADP}.

\section{Theoretical Analysis} \label{sec:theory}
In this section, we first outline the assumptions for our theoretical analysis. Then we demonstrate the convergence upper bound after $T$ rounds. 

\begin{assumption}\label{assumption}
We assume the following conditions for the loss function of all clients:
\begin{enumerate}
    \item $\mathcal{L}_c (M)$ is convex and L-Lipschitz smooth, i.e., for any $M$, $M'$, $\Vert \nabla \mathcal{L}_c (M) - \nabla \mathcal{L}_c (M') \Vert_2 \leq L \Vert M - M' \Vert_2$.
    \item  $\mathcal{L}_c (M)$ satisfies the Polyak-Lojasiewicz condition with the positive parameter $\mu$, such that $\mathcal{L}(M)- \mathcal{L}(M^*) \leq \frac{1}{2 \mu} \Vert \nabla \mathcal{L}(M) \Vert^2_2$, where $M^*$ is the optimal result.
    \item For any $c$ and $M$, $\|\nabla \mathcal{L}_c(M_t) - \nabla \mathcal{L}(M_t)\|_2 \leq \varepsilon_c$ and $\mathbb{E}\{\varepsilon_c\} = \varepsilon$, where $\varepsilon_c$ is the divergence metric.
\end{enumerate}
\end{assumption}



\begin{theorem}\label{theo}
To guarantee $(\epsilon,\delta)$-DP for all clients, the convergence upper bound after $T$ rounds is given by
\begin{equation}
    \begin{aligned}
        &\mathbb{E}\{\mathcal{L}(M_{T})\} - \mathcal{L}(M^*) \leq 
A^T( \mathcal{L}(M_0) - \mathcal{L}(M^*)) \\
&+ L^2 (1 -  A^T)( \frac{S^2 {\ln(1.25/\delta)}}{\mu \epsilon_{\max}^2} + \frac{\eta^2 \varepsilon^2}{2 \mu C}),
    \end{aligned}
\end{equation}
where $A=  1 - 2 \mu \eta + \mu \eta^2 L$.
\end{theorem}

\begin{proof}\label{proof}
First, the client $c$ locally trains the model as: $M_{t+1}^c = M_t - \eta \nabla \mathcal{L}_c(M_t)$. The aggregated model by the server after applying DP can be expressed as: $M_{t+1} = \sum_{c \in \mathcal{C}} p_c \tilde{M}_{t+1}^c = \sum_{c \in \mathcal{C}} p_c (M_{t+1}^c + n_{t+1}^c)$, where $n$ represents the noise, $p_c = \frac{|D_c|}{|D|}$ and $\sum_{c \in \mathcal{C}} p_c =1$. Then, we define $n_{t+1} \triangleq \sum_{c \in \mathcal{C}} p_c n_{t+1}^c$.
By using the second-order Taylor expansion, we have
\begin{equation}
    \begin{aligned}
        &\mathcal{L}(M_{t+1}) - \mathcal{L}(M_t) \\ 
        \leq& (M_{t+1} - M_t)^\top \nabla \mathcal{L}(M_t) + \frac{L}{2} \|M_{t+1} - M_t\|^2_2 \\
         \leq& \left(n_{t+1} - \eta \sum_{c \in \mathcal{C}} p_c \nabla \mathcal{L}_c(M_t)\right)^\top \nabla \mathcal{L}(M_t)\\
& + \frac{L}{2} \left\| \sum_{c \in \mathcal{C}} p_c (n_{t+1}^c - \eta \nabla \mathcal{L}_c(M_t)) \right\|^2_2\\
= &\frac{\eta^2 L}{2} \left\|\sum_{c \in \mathcal{C}} p_c \nabla \mathcal{L}_c(M_t) \right\|^2_2 + \frac{L}{2} \left\|\sum_{c \in \mathcal{C}} p_c n_{t+1}^c \right\|^2_2 \\
& - \eta \nabla \mathcal{L}(M_t)^\top \sum_{c \in \mathcal{C}} p_c \nabla \mathcal{L}_c(M_t).
    \end{aligned}
\end{equation}
Then, the expected objective function is expressed as:
\begin{equation}
    \begin{aligned}
    &\mathbb{E}\{\mathcal{L}(M_{t+1})\} \leq \mathcal{L}(M_t) - \eta \|\nabla \mathcal{L}(M_t)\|^2_2 \\
&+ \frac{\eta^2 L}{2} \mathbb{E}\{\|\sum_{c \in \mathcal{C}} p_c \nabla \mathcal{L}_c(M_t)\|^2_2 \} + \frac{L}{2} \mathbb{E}\{ \|n_{t+1}\|^2_2 \}.
    \end{aligned}
\end{equation}
With the assumption that $p_c=1/C$, the aggregated gradient is:
$\sum_{c \in \mathcal{C}} p_c \nabla \mathcal{L}_c(M_t) = \frac{1}{C} \sum_{c \in \mathcal{C}} \nabla \mathcal{L}_c(M_t)$.
Let $e_c = \nabla \mathcal{L}_c(M_t) - \nabla \mathcal{L}(M_t)$. Thus, according the Assumption \ref{assumption}, $\|\nabla \mathcal{L}_c(M_t) - \nabla \mathcal{L}(M_t)\|_2 = \|e_c\|_2 \leq \varepsilon_c$, which means $e_c$ is bounded by $\varepsilon_c$. Hence, we have $\sum_{c \in \mathcal{C}} p_c \nabla \mathcal{L}_c(M_t)  = \nabla \mathcal{L}(M_t) + \frac{1}{C} \sum_{c \in \mathcal{C}} e_c$. Then, we compute:
\begin{equation}
    \begin{aligned}
        &\mathbb{E}\left\{\left\|\sum_{c \in \mathcal{C}} p_c \nabla \mathcal{L}_c(M_t)\right\|^2_2\right\} \\
        &= \mathbb{E}\left\{ \left\|\nabla \mathcal{L}(M_t) + \frac{1}{C} \sum_{c \in \mathcal{C}} e_c\right\|^2_2 \right\} \\
        &= \mathbb{E}\left\{\|\nabla \mathcal{L}(M_t)\|^2_2 + \frac{2}{C} \nabla \mathcal{L}(M_t)^\top {\sum_{c \in \mathcal{C}} e_c}  + \left\|\frac{\sum_{c \in \mathcal{C}} e_c}{C} \right\|^2_2\right\}.
    \end{aligned}
\end{equation}
Assuming that the $e_c$ are zero-mean, we get $\mathbb{E}\left\{\frac{1}{C} \sum_{c \in \mathcal{C}} e_c\right\} = 0$ and $   \mathbb{E}\left\{\left\|\frac{1}{C} \sum_{c \in \mathcal{C}} e_c\right\|^2_2\right\} \leq \frac{\varepsilon^2}{C}$. Hence, we obtain:
\begin{equation}
    \mathbb{E}\left\{\left\|\sum_{c \in \mathcal{C}} p_c \nabla \mathcal{L}_c(M_t)\right\|^2_2\right\} \leq \|\nabla \mathcal{L}(M_t)\|^2_2 + \frac{\varepsilon^2}{C}.
\end{equation}
Then, we have:
\begin{equation}
    \begin{aligned}
        \mathbb{E}\{\mathcal{L}(M_{t+1})\} \leq & \mathcal{L}(M_t) - \left(\eta - \frac{\eta^2 L}{2}\right) \|\nabla \mathcal{L} (M_t)\|^2_2 \\
        &+ \frac{\eta^2 L}{2} \cdot \frac{\varepsilon^2}{C} + \frac{L}{2} \mathbb{E}\{\|n_{t+1}\|^2_2\}.
    \end{aligned}
\end{equation}
Considering Polyak-Lojasiewicz condition and applying recursively, we know:
\begin{equation}
    \begin{aligned}
    &\mathbb{E}\{\mathcal{L}(M_{t})\} - \mathcal{L}(M^*) \leq 
( 1 - 2 \mu \eta + \mu \eta^2 L)^t( \mathcal{L}(M_0) - \mathcal{L}(M^*)) \\
&+ \frac{L^2 [1 -  ( 1 - 2 \mu \eta + \mu \eta^2 L)^t]}{2 \mu}(\mathbb{E}\{\|n\|^2_2\} + \frac{\eta^2 \varepsilon^2}{C}).
    \end{aligned}
\end{equation}
According to the adaptive DP strategy, we obtain:
\begin{equation}
    \begin{aligned}
        &\mathbb{E}\{\mathcal{L}(M_{T})\} - \mathcal{L}(M^*) \leq 
A^T( \mathcal{L}(M_0) - \mathcal{L}(M^*)) \\
&+ L^2 (1 -  A^T)( \frac{S^2 {\ln(1.25/\delta)}}{\mu \epsilon_{\max}^2} + \frac{\eta^2 \varepsilon^2}{2 \mu C}),
    \end{aligned}
\end{equation}
where $A=  1 - 2 \mu \eta + \mu \eta^2 L$.
\end{proof}
Theorem \ref{theo} establishes a crucial relationship between convergence and privacy in the proposed algorithm. It describes the upper bound on the difference between the model loss $\mathcal{L}(M_T)$ after $T$ rounds and the optimal loss $\mathcal{L}(M^*)$, while ensuring $(\epsilon, \delta)$-DP for all clients. The theorem highlights that the convergence rate of the model is governed by the factor $A$ (where $A < 1$). Additionally, the privacy constraints, represented by $(\epsilon, \delta)$-DP, introduce a trade-off via the noise term. This term illustrates that stricter privacy (lower $\epsilon_{\max}$) leads to slower convergence due to increased noise, emphasizing the need to balance privacy guarantees with efficient convergence.

        



\section{Evaluation}\label{sec:evaluation}
\subsection{Evaluation Setup}
\subsubsection{Implementation Details}
We simulate a FL environment comprising 20 clients, where we employ FedAvg~\cite{mcmahan2017communication} for aggregation. Specifically, we set the local training epochs to 5 and the total global training epochs to 40. Additionally, the learning rate is set to 0.005, the batch size is 64, and the percentage of loss reduction for temporary storage is set to 10\%.
In the FedADP setup, we configure the privacy budget to 3.0, set the global model selection ratio to 60\%, and the model update selection ratio to 70\%. 

\subsubsection{Compared Methods} 
We evaluate FedADP by comparing it with two unlearning strategies: (i) Train-from-scratch (Retrain), which retrains the model from scratch excluding the target clients who request unlearning; (ii) FedRecover~\cite{cao2023fedrecover}, which the server applies the L-BFGS algorithm to realize unlearning (to focus on L-BFGS recovery performance, we did not apply the periodic correction); and (iii) FedADP without DP, which adopts the same selection strategy and unlearning method as FedADP but does not apply differential privacy.

\subsubsection{Evaluation Metrics} 
To assess unlearning effectiveness and privacy preservation, we use \textbf{test accuracy (TA)} and two attacks scenarios. One is the membership inference attack (MIA), introduced by Shokri \textit{et al.} \cite{shokri2017membership}, which aims to determine whether a specific record was used in training a target model. We use the \textbf{membership inference success rate (MISR)} to evaluate the unlearning effectiveness. Specifically, we label training data from target clients as class 1 and data from the test set as class 0. We use the logits output from the compromised global model to train a binary classification model, known as the shadow model. 
After unlearning, we re-input the target clients' training data and obtain the logits output, which is then assessed by the shadow model. The goal is for the results from these target clients to predominantly be classified as 0, indicating failed membership inference and successful removal of their influence from the global model. Since test data is also included, an MISR close to 0.5 indicates better unlearning effectiveness.
Another attack is the backdoor attack (BA) \cite{chen2017targeted,saha2020hidden,bagdasaryan2020backdoor}, which involves injecting malicious triggers into the model during training, causing the model to behave in a targeted way when exposed to specific inputs. 
\textbf{Attack success rate (ASR)} is used to evaluate the performance of backdoor unlearning. We input training data from target clients into the unlearned global model and calculate ASR as the proportion of predictions classified as the target label in a dataset containing the backdoor trigger. A lower ASR indicates a more effective unlearning, as it suggests the model is better at mitigating the backdoor influence.

\begin{figure}[t]
    \centering
    \includegraphics[width=\linewidth]{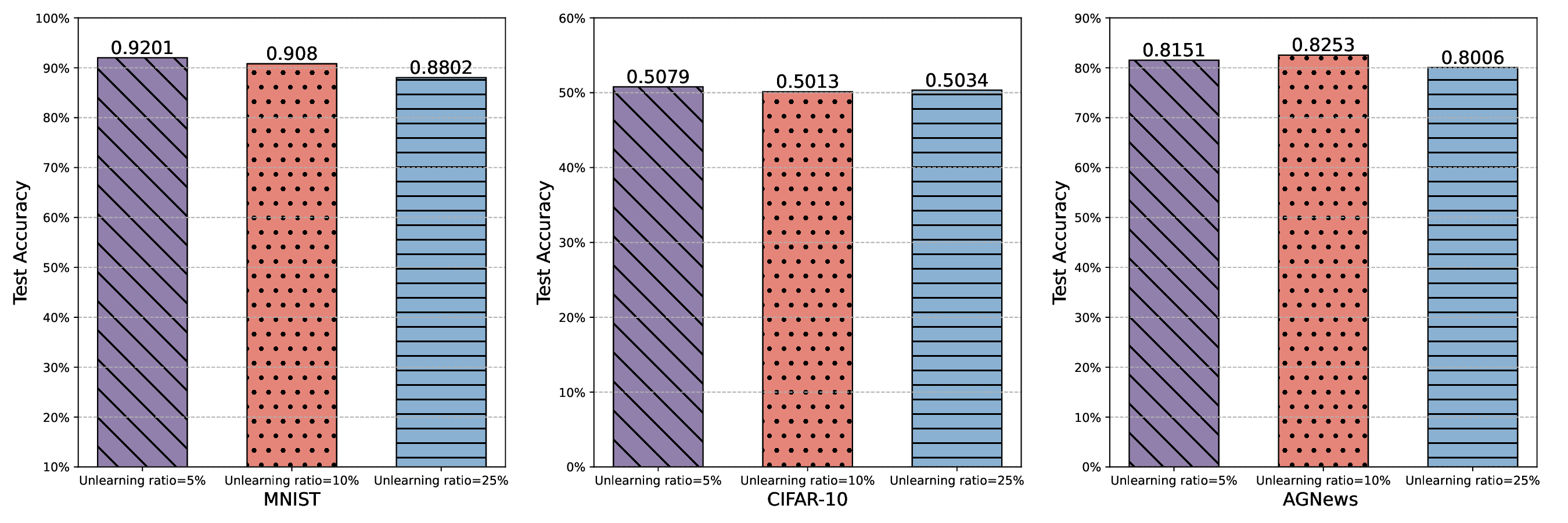}
    \caption{Test accuracy of unlearning ratio of 5\%, 10\% and 25\% on three different datasets of MNIST, CIFAR-10 and AGNews.}
    \label{fig:accuracy_unlearning_ratio}
\end{figure}

\begin{figure}[t]
    \centering
    \includegraphics[width=\linewidth]{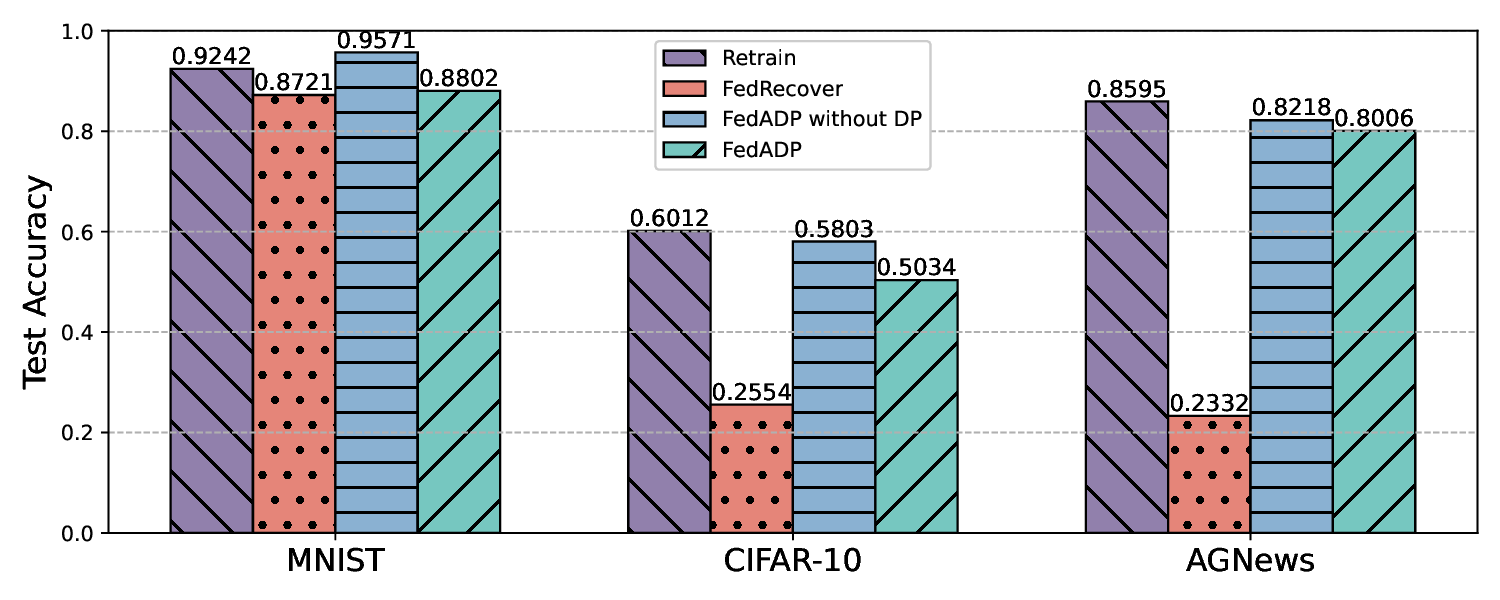}
    \caption{Test accuracy for different methods of Retrain, FedRecover, FedADP without DP and FedADP on dataset MNIST, CIFAR-10 and AGNews.}
    \label{fig:accuracy_methods}
\end{figure}

\begin{figure*}[t]
    \centering
    \includegraphics[width=\linewidth]{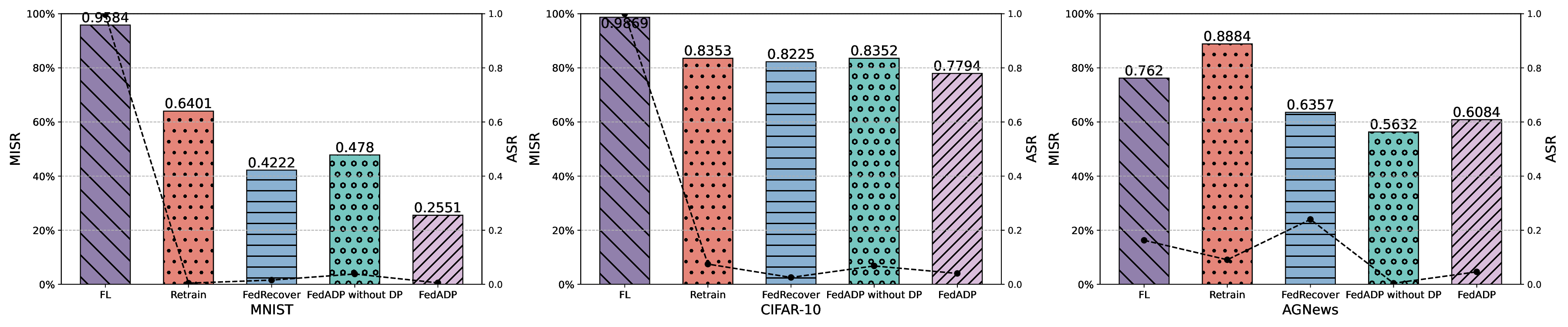}
    \caption{MISR and ASR for different methods of Retrain, FedRecover, FedADP without DP and FedADP on datasets MNIST, CIFAR-10 and AGNews.}
    \label{fig:attacks}
\end{figure*}

\begin{figure}[t]
    \centering
    \includegraphics[width=\linewidth]{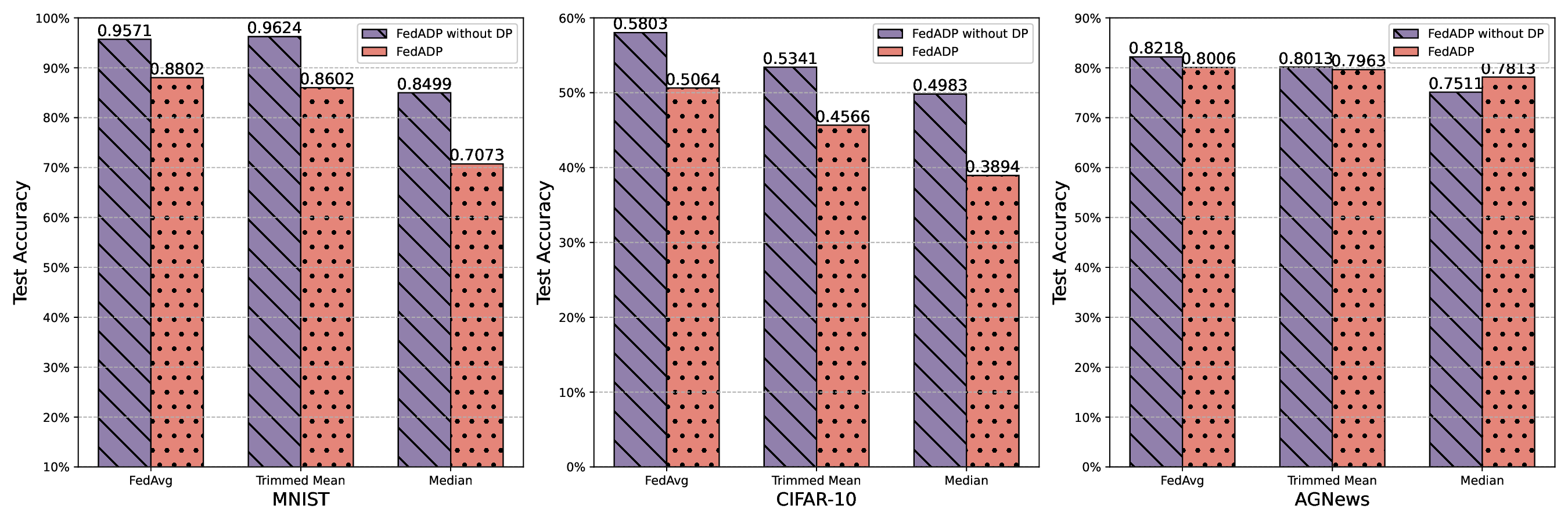}
    \caption{Test accuracy under different secure aggregation methods of FedAvg, Trimmed Mean and Median for different methods of FedADP without DP and FedADP on three different datasets of MNIST, CIFAR-10 and AGNews.}
    \label{fig:aggregation}
\end{figure}

\begin{figure*}[t]
	\centering
	\subfigure[Impact of the model selection ratio]{
		\begin{minipage}{0.31\textwidth} 
            \includegraphics[width=\textwidth]{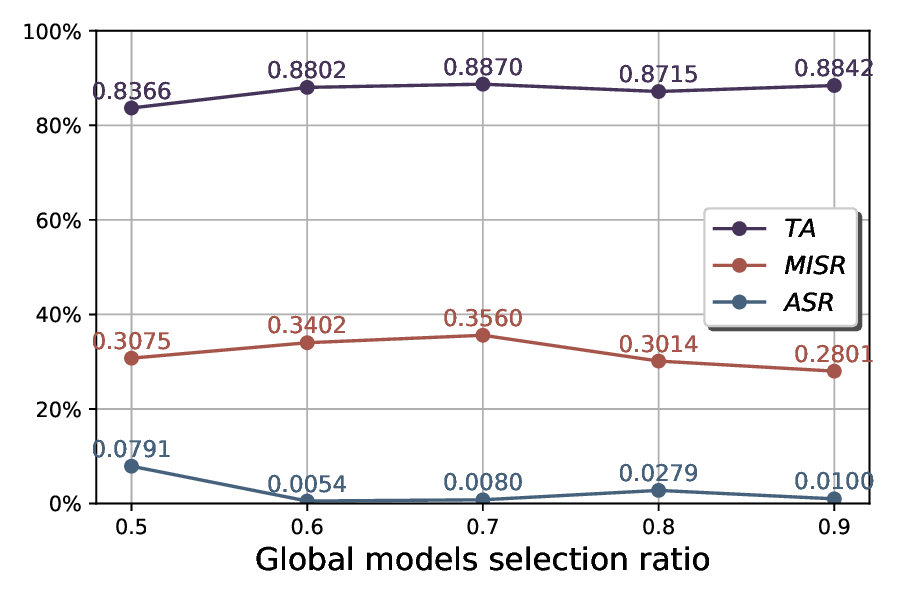} \\
            \label{fig:impact of rounds}
		\end{minipage}
	}
	\subfigure[Impact of the update selection ratio]{
		\begin{minipage}{0.31\textwidth}
			\includegraphics[width=\textwidth]{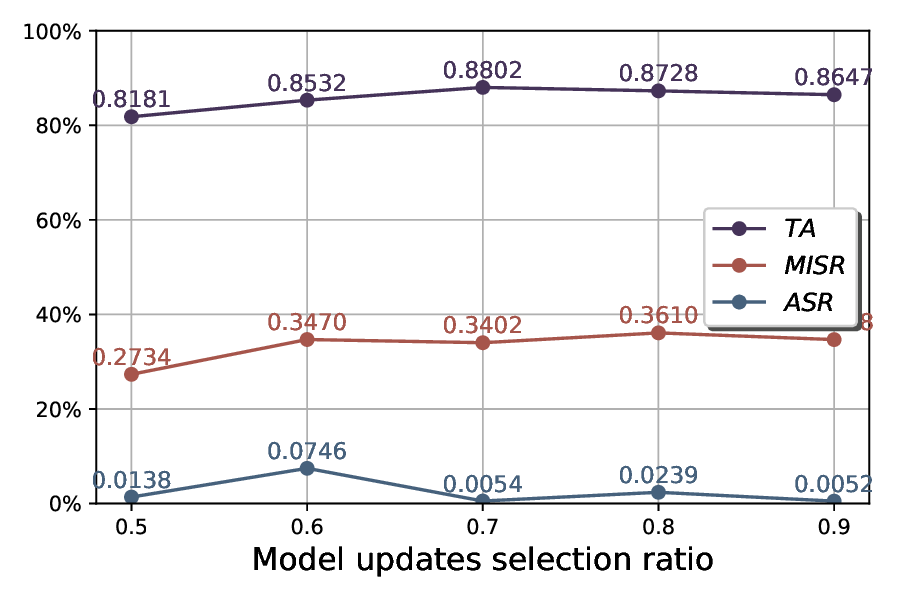}\\
			\label{fig:impact of clients}
		\end{minipage}
	}
 	\subfigure[Impact of the privacy budget]{
		\begin{minipage}{0.31\textwidth}
			\includegraphics[width=\textwidth]{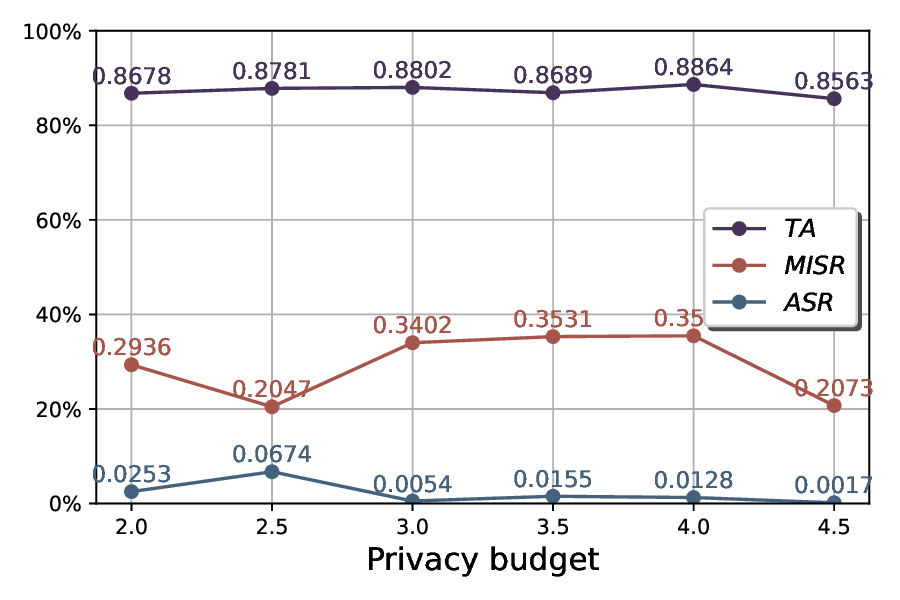}\\
			\label{fig:impact of privacy}
		\end{minipage}
	}
	\caption{Impact of the global model selection ratio, model update selection ratio and privacy budget on MNIST}
    \label{fig:ratio}
\end{figure*}

\subsubsection{Datasets and Models} 

We utilize various datasets and models for classification tasks.
\textbf{MNIST} \cite{deng2012mnist} is used for handwritten digit classification, containing 60,000 training and 10,000 testing images, each $28 \times 28$ pixels. The CNN model includes two convolutional layers with ReLU and max pooling, followed by two fully connected layers transitioning from 1024 to 512 dimensions, and ending with 10 classes.
\textbf{CIFAR-10} \cite{krizhevsky2009learning} features 50,000 training and 10,000 testing images, each $32 \times 32$ pixels, across 10 object and animal classes. The CNN model adapts to three input channels and includes fully connected layers with a dimension of 1600.
\textbf{AGNews} \cite{zhang2015character} comprises 120,000 news articles categorized into World, Sports, Business, and Science/Technology. The TextCNN model uses an embedding layer followed by three 1D convolutional layers with kernel sizes of 3, 4, and 5, along with ReLU activations, max pooling, dropout, and a fully connected layer for class predictions.

\subsection{Evaluation Results and Analysis}

\subsubsection{Storage and Communication Cost}
In this section, we analyze the storage and communication costs associated with each method, as outlined in Table 1. In FL, we assume a total of $T$ training rounds, $C$ clients, and $X$ parameters in the local or global models. Train-from-scratch incurs no storage costs but requires $O(T)$ in communication costs. In contrast, FedRecover demands substantial storage, amounting to $O(CXT)$, and involves communication costs of $O(T_w + T_f + \lfloor \frac{T - T_w - T_f}{T_c} \rfloor)$, where $T_w$ represents the warm-up rounds, $T_c$ the periodic correction rounds, and $T_f$ the final tuning rounds. Our proposed method, FedADP, which incorporates a dual-layer selection strategy involving both global models and model updates, presents a distinct cost profile. Specifically, the storage cost of FedADP is $O(\lambda \gamma CXT)$, while the communication cost is $O(\lambda T)$, where $\lambda < 1$ and $\gamma < 1$. As a result, FedADP offers reductions in both storage and communication costs compared to existing methods.


\begin{table}[!t]
    \centering
    \caption{Comparison of storage and computation cost among methods}
    \begin{tabular}{lll}
     \toprule
         & Storage cost & Computation cost\\
    \midrule
     Retrain    & - & $O(T)$\\
     FedRecover & $O(CXT)$  & $O(T_w+T_f+\lfloor \frac{(T-T_w-T_f)}{T_c} \rfloor$)  \\
     FedADP  & \textbf{$O( \lambda \gamma CXT)$} & \textbf{$O(\lambda T)$}\\
     \bottomrule
    \end{tabular}  
    \label{tab:storage_communication}
\end{table}
\subsubsection{Accuracy on Test Data}
We first analyze the accuracy of FedADP for different unlearning ratios (5\%, 10\%, and 25\%) across three datasets: MNIST, CIFAR-10, and AGNews, as shown in \cref{fig:accuracy_unlearning_ratio}. For the MNIST dataset, although increasing the unlearning ratio (i.e., a higher number of users requesting unlearning) results in a slight decrease in accuracy, it generally remains around 90\%. Given that CIFAR-10 is a larger dataset, the accuracy is somewhat lower compared to MNIST. For the textual dataset AGNews, FedADP achieves an accuracy of over 80\% across different unlearning ratios. This demonstrates FedADP's effectiveness for both image and text datasets.


Additionally, we compare the performance of different methods at an unlearning ratio of 25\% across various datasets, as depicted in \cref{fig:accuracy_methods}. On the MNIST dataset, all four methods achieve accuracy above 87\%, demonstrating robust performance. Among them, FedADP without DP achieves the highest accuracy at 95.71\%, indicating that our selection strategy and calibration method result in superior unlearning performance compared to both Train-from-scratch and FedRecover. FedADP achieves an accuracy of 88.02\%, which is slightly lower due to the added Gaussian noise for privacy protection. While this results in a minor trade-off in accuracy, it ensures client privacy and remains within an acceptable range. Moreover, FedADP's accuracy is higher than that of FedRecover, further demonstrating the effectiveness of FedADP.
On the CIFAR-10 dataset, all methods show a decrease in accuracy, with FedRecover exhibiting the lowest accuracy, highlighting its limitations in handling unlearning for large datasets. 
For the AGNews textual dataset, FedADP without DP and FedADP still show strong unlearning performance.


\subsubsection{Performance after MIA and BA}
We analyze the performance of different methods after MIA and BA on various datasets, as illustrated in \cref{fig:attacks}. MISR and ASR can assess the efficacy of unlearning across different methods. On the MNIST dataset, the MISR is close to 100\% during the FL phase. After applying unlearning methods, all MISR values fall below 50\%, indicating effective unlearning. Notably, FedADP achieves a MISR of 25.51\%, which is significantly lower than the MISR of 47.8\% for FedADP without DP. This demonstrates that incorporating DP not only protects client privacy but also ensures that clients who requested unlearning cannot be identified, validating the effectiveness of FedADP in unlearning.
Similarly, for ASR, the value is also close to 100\% during FL. After applying the four unlearning methods, ASR drops to below 10\%, indicating that the backdoor attack is ineffective against all four methods. Besides, FedADP achieves an exceptionally low ASR of 0.54\%. On the CIFAR-10 dataset, FedADP also exhibits similar performance. For the AGNews text dataset, both FedADP without DP and FedADP outperform existing methods, such as Retrain and FedRecover, in terms of MISR and ASR.


\subsubsection{Generalization across Different Aggregation Methods}
We analyze the test accuracy of FedADP without DP and FedADP with various aggregation methods, such as FedAvg, trimmed mean, and median, across different datasets, as shown in \cref{fig:aggregation}. On the MNIST dataset, FedADP without DP consistently achieves strong unlearning performance across all three aggregation methods, with accuracy exceeding 95\% in both FedAvg and trimmed mean. Although FedADP with DP exhibits slightly lower accuracy due to the added noise, the decrease is not substantial. 
This indicates that our designed selection strategy and calibration method for unlearning maintain excellent performance across different aggregation methods. Despite the drop in effectiveness with DP, the method still demonstrates robust generalization across various aggregation techniques. These findings also apply to the CIFAR-10 and AGNews datasets.

\subsubsection{Discussion}

To analyze the impact of the selection strategy on performance, we conduct two experiments on MNIST focusing on its two main hyperparameters: global model selection ratio $\lambda$ and model update selection ratio $\gamma$. Besides, we adopt adaptive differential privacy, enabling the server to dynamically allocate the privacy budget to clients. This allocation process is guided by a key hyperparameter, maximum privacy budget $\epsilon_{\max}$. In the following experiment, we examine the effects of varying privacy budgets.

\textbf{Impact of the global model selection ratio $\lambda$.}
We selectively store global models during FL to calibrate the subsequent unlearning process. This part explores the impact of varying the global model selection ratio from 0.5 to 0.9 on three key metrics: accuracy, MISR, and ASR, as illustrated in the \cref{fig:impact of rounds}. A higher selection ratio indicates that more global models are stored. 
Our findings reveal that increasing the selection ratio from 0.5 to 0.6 leads to a significant improvement in test accuracy, which then stabilizes at around 88\% as the ratio increases from 0.6 to 0.9. Additionally, the MISR remains consistently below 50\%, indicating a robust unlearning process that effectively minimizes the risk of membership inference attacks. The ASR stays below 10\% across all selection ratios, demonstrating that our approach successfully mitigates backdoor influence. 
Based on these observations, we conclude that a global model selection ratio of 0.6 offers an optimal balance between maintaining model accuracy and ensuring effective attack resilience.

\textbf{Impact of the model update selection ratio $\gamma$.}
Similarly, we discuss the impact of the model update selection ratio from 0.5 to 0.9 on three key metrics, shown in the \cref{fig:impact of clients}. As the update selection ratio increases from 0.5 to 0.7, test accuracy improves noticeably, rising from around 81.81\% to 88.02\%. This suggests that a higher selection ratio within this range effectively enhances model performance by storing more relevant updates. However, beyond a 0.7 selection ratio, accuracy begins to plateau. Meanwhile, the MISR remains relatively stable at approximately 34\%, implying that the model’s resistance to membership inference attacks does not change significantly with different selection ratios. ASR remains low, below 10\%, across all selection ratios, indicating that the unlearning process effectively neutralizes backdoor attacks. 
Hence, we choose 0.7 as the update selection ratio, considering the storage cost.

\textbf{Impact of the privacy budget $\epsilon_{\max}$.}
We examine the impact of varying maximum privacy budgets on model performance during the unlearning process, as shown in the \cref{fig:impact of privacy}. An increase in the privacy budget results in a reduction in the amount of noise added, which reduces privacy preservation. Despite this, the test accuracy at a privacy budget of 3.0 reaches a notable high of 88\%, reflecting strong model performance. Furthermore, when assessing both MISR and ASR across privacy budgets ranging from 2.0 to 4.5, our method exhibits substantial resilience against membership inference attacks and backdoor attacks, demonstrating effective unlearning. Thus, taking into account both privacy preservation and unlearning effectiveness, a privacy budget of 3.0 represents an optimal balance.

\section{Conclusion}\label{sec:conclusions}

In summary, we propose FedADP, an effective and privacy-preserving federated unlearning method that leverages historical information and DP. We implement adaptive DP tailored for FU, which significantly enhances privacy protection. By selecting specific global models and updates for the unlearning process, combined with a novel calibration method, FedADP achieves both efficient and effective unlearning. We theoretically analyze the impact of privacy on convergence and provide an upper bound on convergence while ensuring \((\epsilon, \delta)\)-DP for all clients. Additionally, our experiments analyze the storage and communication cost savings achieved by FedADP, and demonstrate its unlearning effectiveness through metrics such as test accuracy, MISR, and ASR. Furthermore, we explore the impact of the selection strategy and privacy budget on unlearning performance, revealing key trade-offs between privacy protection, storage efficiency, and accuracy. Overall, FedADP provides robust privacy preservation while enabling efficient and accurate unlearning, thereby upholding the ``right to be forgotten" for clients.

\section*{Acknowledgement}
This research / project is supported by the National Research Foundation, Singapore and Infocomm Media Development Authority under its Trust Tech Funding Initiative and the Singapore Ministry of Education Academic Research Fund (RG91/22 and NTU startup). Any opinions, findings and conclusions or recommendations expressed in this material are those of the author(s) and do not reflect the views of National Research Foundation, Singapore and Infocomm Media Development Authority.
\bibliographystyle{IEEEtran}
\bibliography{mybibliography}

\end{document}